\documentclass[conference,a4paper]{IEEEtran}
\usepackage{mleftright}       
\mleftright                   
\usepackage[utf8]{inputenc}
\usepackage[T1]{fontenc}
\usepackage[cmex10]{amsmath}
\usepackage{amssymb}
\usepackage{dsfont}
\usepackage{xcolor}
\usepackage{tabu}
\usepackage{adjustbox}
\usepackage{multirow}
\usepackage{amsthm}
\usepackage{adjustbox}
\usepackage{verbatim}
\usepackage{nicematrix}
\usepackage{tablefootnote}
\usepackage{floatrow}
\usepackage[maxbibnames=99,style=ieee,sorting=nyt,citestyle=numeric-comp]{biblatex}

\usepackage{scrextend}

\makeatletter
\let\MYcaption\@makecaption
\makeatother

\usepackage[font=footnotesize]{subcaption}

\makeatletter
\let\@makecaption\MYcaption
\makeatother

\bibliography{bibliography.bib}

\title{Nearest Neighbor Representations of Neurons}
\author{%
   \IEEEauthorblockN{\textbf{Kordag Mehmet Kilic}\IEEEauthorrefmark{1}, \textbf{Jin Sima}\IEEEauthorrefmark{2} and \textbf{Jehoshua Bruck}\IEEEauthorrefmark{1}}
   \IEEEauthorblockA{\IEEEauthorrefmark{1}%
   Electrical Engineering, California Institute of Technology, USA, \texttt{\{kkilic,bruck\}@caltech.edu}
   }
   \IEEEauthorblockA{\IEEEauthorrefmark{2}%
   Electrical and Computer Engineering, University of Illinois Urbana-Champaign, USA, \texttt{jsima@illinois.edu}
   }
 }

\newtheorem{theorem}{Theorem}
\newtheorem{corollary}{Corollary}[theorem]

\newtheorem{definition}{Definition}

\DeclareMathOperator {\diag}{diag}

\usepackage{tikz}[border=2mm]
\usepackage{tikz-3dplot}
\usepackage{float}
\usetikzlibrary{shapes.geometric,positioning,shapes.gates.logic.US,tikzmark,calc,ext.transformations.mirror}
\AtBeginBibliography{\small}
\tikzset{
    -,
    gate/.style={draw=black,fill=#1,minimum width=5mm,circle},
    square/.style={minimum width=6mm,regular polygon,regular polygon sides=4},
    triangle/.style={minimum width=4mm,regular polygon, regular polygon sides=3},
    every pin edge/.style={draw=black},
    GateCfg/.style={
            logic gate inputs={normal,normal,normal},
            draw,
            scale=2
        }
}

\begin{document}

\maketitle
\thispagestyle{plain}
\pagestyle{plain}
\pagenumbering{gobble}

\begin{abstract}
    The Nearest Neighbor (NN) Representation is an emerging computational model that is inspired by the brain. We study the complexity of representing a neuron (\textit{threshold function}) using the NN representations. It is known that two anchors (the points to which NN is computed) are sufficient for a NN representation of a threshold function, however, the resolution (the maximum number of bits required for the entries of an anchor) is $O(n\log{n})$. In this work, the trade-off between the number of anchors and the resolution of a NN representation of threshold functions is investigated. We prove that the well-known threshold functions EQUALITY, COMPARISON, and ODD-MAX-BIT, which require 2 or 3 anchors and resolution of $O(n)$, can be represented by polynomially large number of anchors in $n$ and $O(\log{n})$ resolution. We conjecture that for all threshold functions, there are NN representations with polynomially large size and logarithmic resolution in $n$.
\end{abstract}

\section{Introduction}
\label{sec:intro}

The human brain has fascinated researchers for many decades in terms of its capabilities and architecture. It consists of a network of simple computational elements, namely, \textit{neurons}. 
A considerable effort has been made in the 1940s and 1950s to model a single neuron biologically and mathematically \cite{hodgkin1952quantitative,mcculloch1943logical,rosenblatt1958perceptron}. The motivation was to model the human brain by constructing networks of the mathematical model of neurons, called \textit{perceptrons}, hence, constructing \textit{neural networks} \cite{minsky2017perceptrons}. After many years of scientific progress, neural networks are demonstrated to be useful in many applications including machine learning, pattern recognition, large language models, and forecasting. It is remarkable that the architecture of the human brain inspired practical solutions to many contemporary challenges and it further motivated many to understand the mathematical limits of these models.

We focus on a relatively new model of the brain called \textit{Nearest Neighbor (NN) Representations} \cite{hajnal2022nearest,kilic2023information}. In this model, one considers concepts as embeddings of vectors in $\mathbb{R}^n$ where arbitrary vectors are explained by the related \textit{nearest neighbors}, called \textit{anchors}, that constitute the NN representation. Even though one can consider each concept as an anchor and perfectly fit the representation without any error, the challenge is to obtain smallest size NN representations. By analogy to language, to present the concept of a `dog', we do not need to represent all dogs. Namely, the goal is to succinctly cluster concepts around the set of anchors. Given the concepts, the complexity of the associated NN representation is defined as the minimal number of anchors.


\begin{figure}[h!]
    \centering
    \begin{floatrow}
    \ffigbox[\FBwidth]
    {
    \scalebox{0.75}{
        \begin{tikzpicture}
        
        \draw[->] (-1,0) -- (3.25,0) node[right] {$X_1$}; 
        \draw[->] (0,-1) -- (0,3.25) node[above] {$X_2$};
        
        \newcommand{\comma}{,}
        \coordinate (c00) at (0,0);
        \coordinate (c01) at (0,2);
        \coordinate (c10) at (2,0);
        \coordinate (c11) at (2,2);
        
        \node[gate,square,color=black,label=225:$(0\comma0)$] (a00) at (0,0) {};
        \node[gate,square,color=black,label=180:$(0\comma1)$] (a01) at (0,2) {};
        \node[gate,square,color=black,,label=270:$(1\comma0)$] (a10) at (2,0) {};
        \node[gate,triangle,color=black,,label=45:$(1\comma1)$] (a11) at (2,2) {};
        
        \draw[dashed] (a01|-a11) -- (a11) -- (a11|-a10);
        
        \node[gate,minimum width=3mm,color=red,label=0:$\textcolor{red}{a_1}$] (anchor-1) at (1,1) {};
        
        \node[gate,minimum width=4.2mm,fill=white,draw=white] (anchor-2) at (2,2) {};
        \node[gate,minimum width=3mm,fill=blue,label=0:$\textcolor{blue}{a_2}$] (anchor-2) at (2,2) {};
        
        \coordinate (cy) at (-0.25,3.25);
        \coordinate (cx) at (3.25,-0.25);
        \draw[-] (cy) -- (cx);
        
        \end{tikzpicture}
}
\label{fig:and_coor}

}

{
\scalebox{0.75}{
        \begin{tikzpicture}
        
        \draw[->] (-1,0) -- (3.25,0) node[right] {$X_1$}; 
        \draw[->] (0,-1) -- (0,3.25) node[above] {$X_2$};
        
        \newcommand{\comma}{,}
        \coordinate (c00) at (0,0);
        \coordinate (c01) at (0,2);
        \coordinate (c10) at (2,0);
        \coordinate (c11) at (2,2);
        
        \node[gate,square,color=black,label=225:$(0\comma0)$] (a00) at (0,0) {};
        \node[gate,triangle,color=black,label=180:$(0\comma1)$] (a01) at (0,2) {};
        \node[gate,triangle,color=black,,label=270:$(1\comma0)$] (a10) at (2,0) {};
        \node[gate,square,color=black,,label=45:$(1\comma1)$] (a11) at (2,2) {};
        
        \draw[dashed] (a01|-a11) -- (a11) -- (a11|-a10);
        
        \node[gate,minimum width=0.5mm,color=red,label=315:$\textcolor{red}{a_1}$] (anchor-1) at (0,0) {};
        
        \node[gate,minimum width=0.5mm,color=blue,label=0:$\textcolor{blue}{a_2}$] (anchor-2) at (1,1) {};
        
        \node[gate,minimum width=0.5mm,color=red,label=0:$\textcolor{red}{a_3}$] (anchor-3) at (2,2) {};
        
        \coordinate (cy) at (-0.25,3.25);
        \coordinate (cx) at (3.25,-0.25);
        \draw (cy) -- (cx);
        \coordinate (dy) at (-1,2);
        \coordinate (dx) at (2,-1);
        \draw (dy) -- (dx);
    \end{tikzpicture}
    }
}
\end{floatrow}
    \caption{NN representations for $2$-input Boolean functions $\text{AND}(X_1,X_2)$ (left) and $\text{XOR}(X_1,X_2)$ (right). Triangles denote $f(X) = 1$ and squares denote $f(X) = 0$. It can be seen that \textbf{\textcolor{red}{red}} anchors are closest to squares and  \textcolor{blue}{blue} anchors are closest to triangles. Separating lines between anchors pairs are drawn.}
    \label{fig:and-xor}
\end{figure}
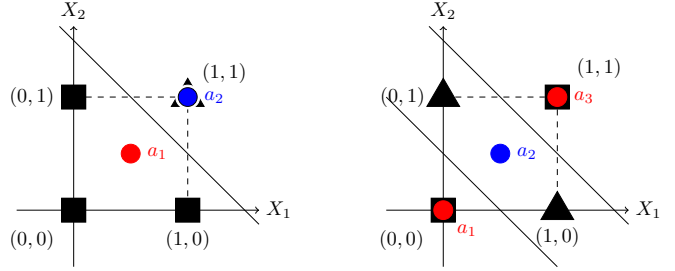

Following the previous works \cite{hajnal2022nearest,kilic2023information}, we focus on the case where each concept is a binary vector $\{0,1\}^n$ and there are only two labels, \textbf{\textcolor{red}{red}} and \textcolor{blue}{blue}. Alternatively, this corresponds to the representations of \textit{Boolean functions}. In this work, we denote real vectors by small letters and binary vectors by capital letters, that is, $x \in \mathbb{R}^n$ or $X \in \{0,1\}^n$ respectively. We also note that anchor entries can be rational numbers. Examples of NN representations for $2$-input Boolean functions AND and XOR are given in Figure \ref{fig:and-xor}.

Now, we introduce some definitions to rigorously represent the NN framework. Let $d(a,b)$ denote the Euclidean distance between two vectors $a,b \in \mathbb{R}^n$. 

\begin{definition}
    The \textup{Nearest Neighbor (NN) Representation} of a Boolean function $f$ is a set of anchors consisting of the disjoint subsets $(P,N)$ of $\mathbb{R}^n$ such that for every $X \in \{0,1\}^n$ with $f(X) = 1$, there exists an anchor $p \in P$ such that for every anchor $n \in N$, $d(X,p) < d(X,n)$, and vice versa. The \textup{size} of the NN representation is $|P \cup N|$.
\end{definition}

\begin{definition}
    The \textup{Nearest Neighbor Complexity} of a Boolean function $f$ is the minimum size over all NN representations of $f$, denoted by $NN(f)$.
\end{definition}

The \textit{resolution of a NN representation} is defined by the maximum number of bits required among the entries of all anchors. This allows us to quantify the \textit{information} required to represent a NN representation of a Boolean function. We write all anchors row-by-row to an \textit{anchor matrix} $A \in \mathbb{Q}^{m\times n}$ for a representation of size $m$.

\begin{definition}
    The \textup{resolution ($RES$) of a rational number} $a/b$ is $RES(a/b) = \lceil \max\{\log_2{|a+1|},\log_2{|b+1|}\}\rceil$ where $a,b \in \mathbb{Z}$, $b \neq 0$, and they are coprime.
    
    For a matrix $A \in \mathbb{Q}^{m\times n}$, $RES(A) = \max_{i,j} RES(a_{ij})$. The \textup{resolution of an NN representation} is $RES(A)$ where $A$ is the corresponding anchor matrix.
\end{definition}

For the XOR function in Figure \ref{fig:and-xor}, we see that the anchors are $\textcolor{red}{a_1} = (0,0)$, $\textcolor{blue}{a_2} = (0.5,0.5)$, and $\textcolor{red}{a_3} = (1,1)$. This representation is in fact optimal in size and the resolution is $\lceil \log_2{3} \rceil = 2$ bits.

Could we represent the brain by NN representations? If so, what would the required resolution be? Namely, we are interested in the NN representations of neural networks under different resolution constraints. In general, such representations may not be practical to analyze albeit possible. Nevertheless, to answer this question, it is fundamental to understand how a single neuron could be represented using nearest neighbors.

For the mathematical model of a neuron, we focus on \textit{threshold functions}. A \textit{linear threshold function} is a weighted summation of binary inputs fed to a step function, namely, $\mathds{1}\{w^T X \geq b\}$ where $w$ is an integer weight vector, $X$ is a binary vector, $b$ is the bias term, and $\mathds{1}\{.\}$ is an indicator function with $\{0,1\}$ outputs. Similarly, an \textit{exact threshold function} is defined as $\mathds{1}\{w^T X = b\}$.

Threshold functions are of interest because they were defined as perceptrons in the classical work on neural networks \cite{minsky2017perceptrons}. When we consider NN representations of threshold functions, we see that $NN(f) = 2$ for linear threshold functions except the constant functions $f(X) = 0$ and $f(X) = 1$, which have $NN(f) = 1$ \cite{hajnal2022nearest,kilic2023information}. Conversely, any two anchor NN representation implies a linear threshold function. In Figure \ref{fig:nn_lt}, the geometrical idea is depicted.

\begin{theorem}[\cite{kilic2023information}]
    \label{th:2-anchor}
    Let $f(X)$ be $n$-input non-constant linear threshold function with weight vector $w \in \mathbb{Z}^n$ and a threshold term $b \in \mathbb{Z}$. Then, there is a $2$-anchor NN representation of $f(X)$ with resolution $O(RES(w))$. In general, the resolution is $O(n\log{n})$.
\end{theorem}

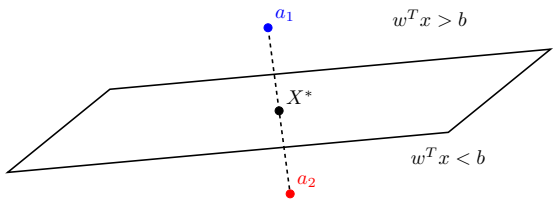
\begin{figure}[h!]
    \centering
    \scalebox{0.75}{
    \tdplotsetmaincoords{70}{110}
        \begin{tikzpicture}[tdplot_main_coords]
            \newcommand{\eq}{=}
            \tdplotsetrotatedcoords{35}{-30}{0}
            \begin{scope}[tdplot_rotated_coords]
            \draw[thick,-] (-4,-4,0) -- (-4,4,0) -- (4,4,0) -- (4,-4,0) -- cycle;
            
            \coordinate (a0) at (0,0,0);

            \node[label=0:$w^T x > b$] (w1) at ($(a0) + (0,2,1.5)$) {};
            \node[label=0:$w^T x < b$] (w3) at ($(a0) + (0,2,-1)$) {};

            \filldraw[black] (a0) circle (2pt) node[above right]  {$X^*$};

            \coordinate (a1) at ($(a0) + (0,0,1.5)$);
            \coordinate (a2) at ($(a0) + (0,0,-1.5)$);
            \draw[thick,-,dash pattern=on 2pt off 2pt](a1)--(a2);
            
            \filldraw[blue] (a1) circle (2pt) node[above right]  {$a_{1}$};
            \filldraw[red] (a2) circle (2pt) node[above right]  {$a_{2}$};
            \end{scope}
        \end{tikzpicture}
    }
    \caption{The NN Representation of a linear threshold function $\mathds{1}\{w^T X \geq b\}$ and its $2$-anchor NN Representation. $X^*$ can be any point in the hyperplane.}
    \label{fig:nn_lt}
\end{figure}

In general, the weights of a threshold function need to be exponentially large, each requiring $O(n\log{n})$ resolution (so that $|w_i| < 2^{O(n\log{n})}$) to be represented \cite{alon1997anti,babai2010weights,haastad1994size,muroga1971threshold}. In contrast, it is reasonable to assert that the brain does not use exponentially ``large'' numbers in its intrinsic network of neurons and to hypothesize that the complex connectivity provides the expressive power. This idea is motivated by the fact that weights in neural networks can be interpreted as the number of synaptic connections among neurons, which cannot be exponentially large. It has been shown that a depth-2 threshold circuit with polynomially large weights in $n$ can be used to compute any threshold function with exponentially large weights in $n$ \cite{amano2005complexity,goldmann1993simulating,hansen2010exact,hofmeister1996note}. We say that the resolution is ``high'' if it is linear in $n$ and that the resolution is ``low'' if it is logarithmic in $n$.

By Theorem \ref{th:2-anchor}, we see that the resolution depends on the weight size of the threshold function, which is exponentially large in $n$. In the light of our hypothesis that the brain does not use ``large'' numbers, we try to prove an analogous result to the result in the context of depth-2 threshold circuits. In other words, our goal is to show that it is possible to increase the number of anchors to a polynomial quantity and decrease the resolution to logarithmic quantities for an arbitrary threshold function. This was posted as an open problem in our previous paper \cite{kilic2023information}. In this paper, we prove that it is true for the threshold functions EQ, COMP, and OMB. However, the general question is still open.

\begin{definition}
    The $2n$-input \textup{EQUALITY} (denoted by $\textup{EQ}_{2n}$) function is an exact threshold function checking if two unsigned $n$-bit integers $X$ and $Y$ are equal. 
    \begin{align}
        \label{eq:equ}
        \textup{EQ}_{2n}(X,Y) &= \mathds{1}\Big\{X = Y\Big\} \\
        \label{eq:equ_equ}
        &= \mathds{1}\Big\{\sum_{i=1}^n 2^{i-1}X_i = \sum_{i=1}^n 2^{i-1}Y_i\Big\}
    \end{align}
    where $X_i$ and $Y_i$ are the binary expansions of $X$ and $Y$.
    
    The $2n$-input \textup{COMPARISON} (denoted by $\textup{COMP}_{2n}$) is the linear threshold function counterpart of the EQ where $\textup{COMP}_{2n}(X,Y) = \mathds{1}\{X \geq Y\}$.
\end{definition}

\begin{definition}
    The $n$-input \textup{ODD-MAX-BIT} (denoted by $\textup{OMB}_n$) function is a linear threshold function checking if the index of the leftmost $1$ of a binary vector $X = (X_1,\dots,X_n)$ has an odd index or not.
    \begin{align}
    \label{eq:omb}
    \textup{OMB}_n(X) &= \mathds{1}\Big\{\sum_{i=1}^n (-1)^{i-1} 2^{n-i} X_i > 0 \Big\}
\end{align}
\end{definition}

We first show that any exact threshold function has at most a $3$-anchor NN representation with resolution $O(n\log{n})$ to conclude the treatment of the NN representations of threshold functions with arbitrary resolution. Then, we give constructions of the low resolution NN representations for the EQ, COMP, and OMB. We believe that our results could provide insights to solve the general open problem, namely, if an arbitrary threshold function has NN representations with polynomially large number of anchors and logarithmic resolution in $n$. Table \ref{tab:nn} presents the summary of our results.

\begin{table}[h!]
\caption{The Results for Low Resolution NN Representations}
\label{tab:nn}
\begin{minipage}{\textwidth}
\centering
\begin{adjustbox}{width=0.65\textwidth}
\begin{tabular}{|c||c||c|}
    \hline
     \multirow{2}{*}{\textbf{Function}} & \multicolumn{2}{c|}{\textbf{NN Representation}} \\ 
     \cline{2-3} & Size & Resolution \\
     \hline\hline
     \hline
     $\text{EQ}_{2n}$ & $2n + 1$ & $O(1)$ \\
     \hline
     $\text{EQ}_{2n}$ & $O(n/\log{n})$ & $O(\log{n})$ \\
     \hline
     $\text{COMP}_{2n}$ & $2n$ & $O(\log{n})$ \\
     \hline
     $\text{OMB}_{n}$ & $n+1$ & $O(\log{n})$\\
     \hline
\end{tabular}
\end{adjustbox}
\end{minipage}
\end{table}

\section{NN Representations of Threshold Functions with High Resolution}
\label{sec:lt}

In addition to the result in Theorem \ref{th:2-anchor}, we give a similar result for exact threshold functions for completeness. For any exact threshold function $f$ which is not a linear threshold function, it is $\text{NN}(f) > 2$ necessarily. For example, the EQ is such an example and the AND function is both exact and linear. We show that any exact threshold function which is not linear has exactly a $3$-anchor NN representation where all the anchors are collinear. The geometrical idea is given in Figure \ref{fig:nn_elt}. Contrary to linear threshold functions, the converse does not hold, i.e, a Boolean function with a $3$-anchor NN representation need not to be an exact threshold function (see \textit{$3$-interval symmetric Boolean functions} in \cite{kilic2023information}). The resolution of the weights is related to the resolution of the anchors similar to the result in Theorem \ref{th:2-anchor}.

\begin{theorem}
    \label{th:3-anchor}
    Let $f(X)$ be an $n$-input exact threshold function with weight vector $w \in \mathbb{Z}^n$ and the bias term $b \in \mathbb{Z}$. Then, there is a $3$-anchor NN representation of $f(X)$ with resolution $O(RES(||w||_2^2))$ with $||.||_2$ being the Euclidean norm. In general, the resolution is $O(n\log{n})$.
\end{theorem}
\begin{proof}
    We follow the idea in the proof of Theorem \ref{th:2-anchor}. In addition to $a_1$ and $a_2$, we put another anchor $a_0$ which a solution to $w^T X = b$. $a_0$ have the opposite labeling of the $a_1$ and $a_2$. We assume that $c > 0$.
    \begin{align}
        \textcolor{blue}{a_0} &= X^*\\
        \textcolor{red}{a_1} &= X^* - cw \\
        \textcolor{red}{a_2} &= X^* + cw
    \end{align}
    Without loss of generality, we can assume that there exist a binary $X^*$ such that $w^T X^* = b$ because otherwise, $f(X) = 0$ for all $X \in \{0,1\}^n$, which is a constant function with a trivial $1$-anchor NN representation. We have the following necessary and sufficient conditions to claim that this representation is valid indeed:
    \begin{align}
        \textbf{\underline{Case 1}: } w^T X &= b \Leftrightarrow d(a_0,X) < d(a_i,X) \text{ for } i = 1,2\\
        \textbf{\underline{Case 2}: } w^T X &\neq b \Leftrightarrow d(a_0,X) > d(a_i,X) \text{ either for } i = 1,2
    \end{align}
    We get the following when we expand the squared Euclidean distance from the input vector to an anchor:
    \begin{align}
        d(a_1,X)^2 &= |X| - 2a_1^T X + ||a_1||_2^2 \\ \nonumber
                &= |X| - 2X^T X^* + ||X^*||_2^2 \\ 
                &\hphantom{aaaa}+ 2c(w^T X - b) + c^2 ||w||_2^2\\
                &= d(a_0,X)^2 + 2c(w^T X - b) + c^2 ||w||_2^2 \\
        d(a_2,X)^2 &= d(a_0,X)^2 - 2c(w^T X - b) + c^2 ||w||_2^2
    \end{align}
    It is clear that $\min_{i \in \{1,2\}} d(a_i,X)^2 = d(a_0,X)^2 - 2c|w^T X - b| + c^2 ||w||_2^2$. When we write the necessary and sufficient conditions more explicitly, we obtain
    \begin{align}
        \textbf{\underline{Case 1}: } w^T X &= b \Leftrightarrow 0 < -2c|w^T X - b| + c^2 ||w||_2^2\\
        \textbf{\underline{Case 2}: }  w^T X &\neq b \Leftrightarrow 0 > -2c|w^T X - b| + c^2 ||w||_2^2
    \end{align}
    Case 1 is trivial as long as $c \neq 0$ because $w^T X = b$. For Case 2, we have $c < \frac{2|b - w^T X|}{||w||_2^2}$. The minimum value the numerator can take is $1$, therefore, the bound is tightest when $c < \frac{2}{||w||_2^2}$. Taking $c = \frac{1}{||w||_2^2}$ suffices and we immediately see that $RES(A) = O(RES(||w||_2^2))$ given that $X^*$ is binary. In \cite{babai2010weights}, since $RES(w) = O(n\log{n})$, we obtain the $RES(A) = O(n\log{n})$.
\end{proof}

Compared to Theorem \ref{th:2-anchor}, the resolution is upper bounded by the norm of the weights in Theorem \ref{th:3-anchor}. When the weights are constant in $n$, the norm of the weights might be linear in $n$ and the resolution may not be constant anymore. Geometrically speaking, this is due to the fact that if the anchors $a_1$ and $a_2$ are too far from the hyperplane, $a_0$ might get closer to binary vectors outside the hyperplane itself.

\begin{figure}[h!]
    \centering
    \scalebox{0.75}{
    \tdplotsetmaincoords{70}{110}
        \begin{tikzpicture}[tdplot_main_coords]
            \newcommand{\eq}{=}
            \tdplotsetrotatedcoords{35}{-30}{0}
            \begin{scope}[tdplot_rotated_coords]
            \draw[thick,-] (-4,-4,0) -- (-4,4,0) -- (4,4,0) -- (4,-4,0) -- cycle;
            
            \coordinate (a0) at (0,0,0);

            \node[label=0:$w^T x > b$] (w1) at ($(a0) + (0,2,1.5)$) {};
            \node[label=0:$w^T x \eq b$] (w2) at ($(a0) + (0,2,0.25)$) {};
            \node[label=0:$w^T x < b$] (w3) at ($(a0) + (0,2,-1)$) {};

            \coordinate (a1) at ($(a0) + (0,0,1.5)$);
            \coordinate (a2) at ($(a0) + (0,0,-1.5)$);
            \draw[thick,-,dash pattern=on 2pt off 2pt](a1)--(a2);
            
            \filldraw[red] (a1) circle (2pt) node[above right]  {$a_{1}$};
            \filldraw[red] (a2) circle (2pt) node[above right]  {$a_{2}$};

            \filldraw[blue] (a0) circle (2pt) node[above right]  {$a_0$};
            \end{scope}
        \end{tikzpicture}
    }
    \caption{The NN Representation of an Exact Threshold Function $\mathds{1}\{w^T X = b\}$ and its $3$-anchor NN Representation. The anchors $a_1$ and $a_2$ must be close enough to the hyperplane. All anchors are collinear.}
    \label{fig:nn_elt}
\end{figure}
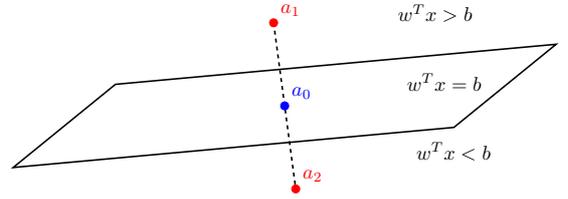

\section{Low Resolution NN Representations for EQ}

A circuit theoretic result shows that there are constant weight \text{EQ matrices} that can be used to compute the EQ function to optimize the number of the threshold gates \cite{kilic2022algebraic}. We now show that one can use any EQ matrix to reduce the number of anchors in a similar way. In our analysis, we use $\mathds{1}$ to denote the all-one vector, $||.||_2$ to denote Euclidean distance, and $\diag(AA^T) \in \mathbb{Z}^m$ to denote squared Euclidean norm of each row of an anchor matrix $A \in \mathbb{Z}^{m\times n}$.

\begin{definition}
A matrix $A \in \mathbb{Z}^{m\times n}$ is an $\textup{EQ matrix}$ if the homogeneous system $Ax = 0$ has no non-trivial solutions in $\{-1,0,1\}^n$.
\end{definition}

\begin{theorem}
    \label{th:eq_log}
    Consider an EQ matrix $A \in \mathbb{Z}^{m\times n}$ with no all-zero rows. Then, there is a $2m+1$-anchor NN representation for the $2n$-input \textup{EQ} function with $O(RES(\diag(AA^T)))$
\end{theorem}
\begin{proof}
    The anchors we construct are in the following form for $i \in \{1,\dots,m\}$. We define $W = \begin{bmatrix} A & -A\end{bmatrix}_{m\times 2n}$.
    \begin{align}
        \textcolor{blue}{a_0} &= 0.5\mathds{1} \\
        \textcolor{red}{a_{2i-1}} &= 0.5\mathds{1} + c_iW_i \\
        \textcolor{red}{a_{2i}} &= 0.5\mathds{1} - c_iW_i
    \end{align}
    We have two cases similar to the proof of Theorem \ref{th:3-anchor} to prove where $a_0$ is closest to vectors $X = Y$ and $a_i$ is closest to $X \neq Y$ for some $i \in \{1,\dots,2n\}$. We also use $(X,Y) \in \{0,1\}^{2n}$ to denote the input vector. When we expand the squared Euclidean distance, we get
    \begin{align}
        d(a_0,(X,Y))^2 &= n/2 \\
        d(a_{2i-1},(X,Y))^2 &= n/2 - 2c_iA_i^T(X-Y) + 2c_i^2||A_i||_2^2 \\
        d(a_{2i},(X,Y))^2 &= n/2 + 2c_iA_i^T(X-Y) + 2c_i^2||A_i||_2^2
    \end{align}
    Since $X-Y \in \{-1,0,1\}^n$, $A(X-Y) = 0$ if and only if $X = Y$. Therefore, if $X=Y$, $d(a_0,(X,Y))^2 < d(a_i,(X,Y))^2$ for all $i \in \{1,\dots,2n\}$ as long as $c_i > 0$.
    
    Conversely, if $X \neq Y$, then $|A_i^T(X-Y)| \geq 1$ for some $i \in \{1,\dots,n\}$. Depending on the sign of $A_i^T(X-Y)$, one of the $a_{2i-1}$ or $a_{2i}$ will be closer to $(X,Y)$. Hence, if $c_i = \frac{1}{2||A_i||_2^2}$, then
    \begin{align}
        -2c_i|A_i^T(X-Y)| + 2c_i^2||A_i||_2^2 = -\frac{|A_i^T(X-Y)|-0.5}{||A_i||_2^2} < 0
    \end{align}
    so that either $a_{2i-1}$ or $a_{i}$ will be closest to $(X,Y)$. The resolution bound is easy to verify by the selection of $c_i$s.
\end{proof}

The properties of EQ matrices directly translate into the NN representations of the EQ. We obtain the following corollary concluding our results in this context.

\begin{corollary}
    \label{cor:eq}
    For the $2n$-input EQ function, there are NN representations such that
    \begin{table}[h!]
        \centering
        \begin{tabular}{|c|c|}
            \hline
            Size & Resolution \\
            \hline
            3 & $O(n)$ \\
            \hline
            $2n+1$ & $O(1)$ \\
            \hline
            $O(n/\log{n})$ & $O(\log{n})$ \\
            \hline
        \end{tabular}
    \end{table}
\end{corollary}
\begin{proof}
    For the first result, we use the EQ matrix with exponentially large weights (see Eq. \eqref{eq:equ_equ}) taking $m = 1$. This gives the linear resolution construction. For the second result, we use $A = I$, which is a full-rank matrix for the EQ function where $m = n$. Since each row has norm 1, this gives $O(1)$ resolution.
    Finally, for the third result, we use a result on EQ matrices in \cite{kilic2022algebraic} (see Theorem 1). In this case, the number of rows is asymptotically $O(n/\log{n})$ and each row norm is bounded by $n$, giving a logarithmic resolution.
\end{proof}

\section{Low Resolution NN Representations for COMP and OMB}

One can notice that for the COMP and OMB, the most significant bits  in the input determine the output. This is called the \textit{domination principle} \cite{kilic2021neural} and the property used implicitly in many works \cite{amano2005complexity,bohossian1998trading,kilic2022algebraic}. In order to decide the label of the anchor using the domination principle, a procedure based on a sequential decision making could be useful similar to decision lists. Therefore, loosely speaking, some anchors should be ``closer'' more often than the others.

One can see the following necessary and sufficient conditions for the COMP with the convention $X = (X_1,\dots,X_n)$ and $Y = (Y_1,\dots,Y_n)$ and binary expansion $\sum_{i=1}^n 2^{n-i}X_i$. Below, $\times$s denote \textit{don't cares}.
\begin{align}
    \label{eq:comp_1}
    X > Y &\Leftrightarrow X-Y = (0,\dots,0,\hphantom{l}1,\times,\dots,\times) \\
    X < Y &\Leftrightarrow  X-Y = (0,\dots,0,-1,\times,\dots,\times) \\
    X = Y &\Leftrightarrow  X-Y = (0,\dots,0,0,0,\dots,0)
\end{align}

The vector $X-Y$ has leading $0$s and the most significant digit determines if $X > Y$ or not. The idea to construct an NN representation for the COMP relies on this observation and we depict the idea in Figure \ref{fig:comp}. First, we consider the hyperplane $x_1 = y_1$ and put two anchors whose midpoint is $X^* = 0.5\mathds{1}$, which is on the hyperplane itself so that $a_1$ is closer to $X_1 > Y_1$ and $a_2$ is closer to $X_1 < Y_1$. Then, alongside the intersection of $x_1 = y_1$ and $x_2 = y_2$, we put two more anchors $a_3$ and $a_4$ to take care of $X_2 > Y_2$ and $X_2 < Y_2$ and we continue like this. To ensure that we have the domination principle, $a_3$ and $a_4$ are farther away from $X^*$. Since $\text{COMP}(X,Y) = 1$ when $X = Y$, we make the anchors skewed by a little amount so that $X = Y$ vectors are closer to \textcolor{blue}{blue} anchors.

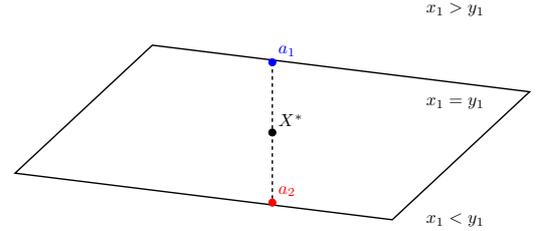
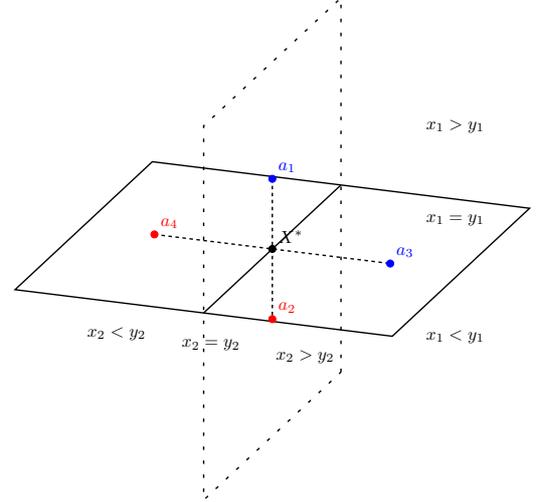
\begin{figure}[h!]
    \centering
    \begin{subfigure}{\textwidth}
    \centering
    \begin{adjustbox}{trim=0cm 2cm 0cm 1cm,clip}
    \scalebox{0.66}{
        \tdplotsetmaincoords{70}{110}
        \begin{tikzpicture}[tdplot_main_coords]
            \newcommand{\eq}{=}
            \draw[thick,-,dash pattern=on 2pt off 6pt,draw=white] (-4,0,-4) -- (-4,0,4) -- (4,0,4) -- (4,0,-4) -- cycle;
            \draw[thick,-] (-4,-4,0) -- (-4,4,0) -- (4,4,0) -- (4,-4,0) -- cycle;
            
            \coordinate (a0) at (0,0,0);

            \node[label=0:$x_1 > y_1$] (w1) at ($(a0) + (0,3,3)$) {};
            \node[label=0:$x_1 \eq y_1$] (w2) at ($(a0) + (0,3,1)$) {};
            \node[label=0:$x_1 < y_1$] (w3) at ($(a0) + (0,3,-1.5)$) {};
            
            \coordinate (X) at ($(a0)$);

            \filldraw[black] (X) circle (2pt) node[above right]  {$X^*$};

            \coordinate (a21) at ($(a0) + (0,0,1.5)$);
            
            \coordinate (a22) at ($(a0) + (0,0,-1.5)$);
            \draw[thick,-,dash pattern=on 2pt off 2pt](a21)--(a22);
            
            \filldraw[blue] (a21) circle (2pt) node[above right]  {$a_1$};
            
            \filldraw[red] (a22) circle (2pt) node[above right]  {$a_2$};
        \end{tikzpicture}
    }
    \end{adjustbox}
    \caption{The placement of anchors $a_1$ and $a_2$ for the hyperplane $x_1 = y_1$. $X^* = 0.5\mathds{1}$ for simplicity.}
    \label{sub:comp_idea_2}
    \end{subfigure}
    
    \begin{subfigure}{\textwidth}
    \centering
    \scalebox{0.66}{
        \tdplotsetmaincoords{70}{110}
        \begin{tikzpicture}[tdplot_main_coords]
            \newcommand{\eq}{=}
            \draw[thick,-,dash pattern=on 2pt off 6pt] (-4,0,-4) -- (-4,0,4) -- (4,0,4) -- (4,0,-4) -- cycle;
            \draw[thick,-] (-4,-4,0) -- (-4,4,0) -- (4,4,0) -- (4,-4,0) -- cycle;
            \draw[thick,-](-4,0,0)--(4,0,0);
            
            \coordinate (a0) at (0,0,0);

            \node[label=0:$x_1 > y_1$] (w1) at ($(a0) + (0,3,3)$) {};
            \node[label=0:$x_1 \eq y_1$] (w2) at ($(a0) + (0,3,1)$) {};
            \node[label=0:$x_1 < y_1$] (w3) at ($(a0) + (0,3,-1.5)$) {};

            \node[label=0:$x_2 > y_2$] (w4) at ($(a0) + (6,2,0)$) {};
            \node[label=0:$x_2 \eq y_2$] (w5) at ($(a0) + (6,0,0)$) {};
            \node[label=0:$x_2 < y_2$] (w6) at ($(a0) + (6,-2,0)$) {};
            
            \coordinate (X) at ($(a0)$);

            \filldraw[black] (X) circle (2pt) node[above right]  {$X^*$};

            \coordinate (a11) at ($(a0) + (0,2.5,0)$);
            
            \coordinate (a12) at ($(a0) + (0,-2.5,0)$);
            
            \draw[thick,-,dash pattern=on 2pt off 2pt](a12)--(a11);

            \coordinate (a21) at ($(a0) + (0,0,1.5)$);
            
            \coordinate (a22) at ($(a0) + (0,0,-1.5)$);
            \draw[thick,-,dash pattern=on 2pt off 2pt](a21)--(a22);
            
            \filldraw[blue] (a11) circle (2pt) node[above right]  {$a_3$};
            
            \filldraw[red] (a12) circle (2pt) node[above right]  {$a_4$};
            
            \filldraw[blue] (a21) circle (2pt) node[above right]  {$a_1$};
            
            \filldraw[red] (a22) circle (2pt) node[above right]  {$a_2$};
        \end{tikzpicture}
    }
    \caption{The placement of $a_3$ and $a_4$ for the hyperplanes $x_1 = y_1$ and $x_2 = y_2$. They are farther away from the $X^*$ compared to $a_1$ and $a_2$ to ensure the domination principle.}
    \label{sub:comp_idea_4}
    \end{subfigure}
    \caption{The construction idea for the $\text{COMP}(X,Y)$ function depicting the first two iterations. For $2n$-inputs, there will be $n$ iterations resulting in $2n$ many anchors.}
    \label{fig:comp}
\end{figure}

\begin{theorem}
    \label{th:comp}
    For the $2n$-input \textup{COMP}, there is an NN representation with $2n$ anchors and $O(\log{n})$ resolution.
\end{theorem}
\begin{proof}
    We will show that the following is an NN representation for the $\text{COMP}_{2n}$ function for $i \in \{1,\dots,n\}$ and $W = \begin{bmatrix} I & -I\end{bmatrix}_{n\times 2n}$
    \begin{align}
        X^* &= 0.5\mathds{1} \\
        \textcolor{blue}{a_{2i-1}} &= X^* + c_{2i-1}W_i \\
        \textcolor{red}{a_{2i}} &= X^* - c_{2i}W_i
    \end{align}
    for $c_i = \frac{1}{2} + \frac{i-1}{4n}$ for $i \in \{1,\dots,2n\}$. This selection gives the desired resolution bound.

    Let $(X,Y)$ denote the $2n$-dimensional input vector and let $(\mathcal{X},\mathcal{Y})^{(k)}$ be the set of vectors where $X_i = Y_i$ for $i < k$, $X_k = 1$ and $Y_k = 0$. For $i > k$, the values can be arbitrary. Clearly, $X > Y$ for any vector in $(X,Y)^{(k)} \in (\mathcal{X},\mathcal{Y})^{(k)}$ (see Eq. \eqref{eq:comp_1}). We claim that the closest anchor to any $(X,Y)^{(k)}$ is $a_{2k-1}$.
    \begin{align}
        \nonumber
        &d(a_{2i-1},(X,Y)^{(k)})^2 \\
        \nonumber
        &\hphantom{aa}=|(X,Y)^{(k)}| -2(0.5\mathds{1}+c_{2i-1}W_i)^T (X,Y)^{(k)}\\
        \nonumber
        &\hphantom{aaaa}+ ||0.5\mathds{1}+c_{2i-1}W_i||_2^2\\
        \nonumber
        &\hphantom{aa}= -2c_{2i-1}(X_i - Y_i) + \frac{n}{2} + 2c_{2i-1}^2||W_i||_2^2 \\
        &\hphantom{aa}= \frac{n}{2} + \frac{1}{2} + \frac{i-1}{n} + 2\Big(\frac{i-1}{2n}\Big)^2 -\Big(1+\frac{i-1}{n}\Big)(X_i - Y_i)
    \end{align}
    Since $-\frac{1}{2} + 2\Big(\frac{k-1}{2n}\Big)^2 < 0 < \frac{1}{2} + \frac{i-1}{n} + 2\Big(\frac{i-1}{2n}\Big)^2$ for any $i < k$, $a_{2k-1}$ is closer to $(X,Y)^{(k)}$ than $a_{2i-1}$ for $i < k$. For $i > k$, the smallest distance value is attained when $X_i - Y_i = 1$. For this, we see that $-\frac{1}{2} + 2\Big(\frac{k-1}{2n}\Big)^2 < -\frac{1}{2} + 2\Big(\frac{i-1}{2n}\Big)^2$ for $i > k$ since this expression is monotonically increasing.

    We now look at the distance between $(X,Y)^{(k)}$ and $a_{2i}$s.
    \begin{align}
        \nonumber
        d(a_{2i},(X,Y)^{(k)})^2 &= \frac{n}{2} + \frac{1}{2} + \frac{2i-1}{2n} + 2\Big(\frac{2i-1}{4n}\Big)^2 \\
        &\hphantom{aaa} +\Big(1+\frac{2i-1}{2n}\Big)(x_i - y_i)
    \end{align}
    When we compare $d(a_{2i-1},(X,Y)^{(k)})^2$ and $d(a_{2i},(X,Y)^{(k)})^2$ for $i < k$, we have $-\frac{1}{2} + 2\Big(\frac{k-1}{2n}\Big)^2 < 0 < \frac{1}{2} + \frac{2i-1}{2n} + 2\Big(\frac{2i-1}{4n}\Big)^2$. Also, for $i = k$, $d(a_{2k-1},(X,Y)^{(k)}) < d(a_{2k},(X,Y)^{(k)})$ because $x_k - y_k > 0$. For $i > k$, the minimum distance value is attained when $x_i - y_i = -1$, so we obtain $-\frac{1}{2} + \Big(\frac{2k-2}{4n}\Big)^2 < -\frac{1}{2} + \Big(\frac{2i-1}{4n}\Big)^2$. Hence, we conclude that $a_{2i-1}$ is closer to $(X,Y)^{(k)}$ than $a_{2i}$ for any $i$.
    
    The proof is similar for $X < Y$. For $X = Y$, we see that $a_1$ is always the closest anchor with the correct label because $c_1$ is the smallest. Hence, the construction works indeed.
\end{proof}

We finish this section with the construction for the OMB. Note that in this context, we use $X = (X_1,\dots,X_n)$ with $X_1$ being the most significant bit.

\begin{theorem}
    \label{th:omb}
    For the $n$-input \textup{OMB}, there is an NN representation with $n+1$ anchors with $O(\log{n})$ resolution.
\end{theorem}
\begin{proof}
    The construction we give here is very simple, which can be thought as a special case of a wider family of constructions that we do not cover here. We have the anchor matrix $A \in \mathbb{Q}^{(n+1)\times n}$ with resolution $O(\log{n})$ for $i \in \{1,\dots,n+1\}$ and $j \in \{1,\dots,n\}$.
    \begin{align}
        a_{ij} = \begin{cases}
            1-\frac{i-1}{n} &\text{ if } i = j \\
            0 &\text{ otherwise}
        \end{cases}
    \end{align}
    The label of $a_i$ is \textcolor{blue}{blue} for odd $i$ and \textcolor{red}{\textbf{red}} for even $i$. Moreover, the last anchor $a_{n+1}$ corresponds to the all-zero vector and it is labeled \textcolor{red}{\textbf{red}} because $\text{OMB}(0) = 0$ by definition.
    
    By $\mathcal{X}^{(k)}$, we denote the family of vectors where $X_j = 0$ necessarily for $j < k$ and $x_k = 1$. That is, for any $X \in \mathcal{X}^{(k)}$ $X = (0,\dots,0,1,\times,\dots,\times)$. For $j > k$, we have \textit{don't cares}.

    We claim that all $X \in \mathcal{X}^{(k)}$ are closest to $a_i$. By expanding the squared Euclidean distance, we have
    \begin{align}
        &d(a_i,X)^2 = |X| - 2a_i^T X + ||a_i||^2 \\
         &\hphantom{aa}= \begin{cases}
            |X| + \Big(1 - \frac{i-1}{n}\Big)^2 &\text{ if } i < k \\
            |X| - 2\Big(1 - \frac{i-1}{n}\Big)X_{i} + \Big(1 - \frac{i-1}{n}\Big)^2 &\text{ if } k \leq i \leq n\\
            |X| &\text{ if } i = n+1
        \end{cases}
    \end{align}
    Since $X_k = 1$ and $-2\Big(1 - \frac{i-1}{n}\Big) + \Big(1 - \frac{i-1}{n}\Big)^2 < 0 < \Big(1 - \frac{i-1}{n}\Big)^2$, $a_k$ is closer to $X$ than $a_i$s for $i < k$. Since $-2\Big(1 - \frac{i-1}{n}\Big) + \Big(1 - \frac{i-1}{n}\Big)^2 = -1 + \Big(\frac{i-1}{n}\Big)^2$ is monotone increasing in $i$ and always negative, the closest anchor is $a_k$ indeed. This concludes the proof.
\end{proof}

We remark that the inclusion of the all-zero anchor for the NN representation of OMB is not necessary when $n$ is even because the label of $a_n$ and $a_{n+1}$. In this case, the representation size is exactly $n$.

\section{Conclusion}

NN representations are promising models in the understanding of how the concepts in the brain are represented. In this regard, NN representations of threshold functions are analyzed and the trade-off between the resolution of the anchors and the number of anchors is studied. The main contributions in the paper include characterizing the NN representations of exact threshold functions as well as polynomial size constructions of NN representations with low resolution for the EQ, COMP and OMB functions. In contrast, these functions require high resolution with a constant number of anchors. Whether low resolution NN representations for any threshold function with polynomially large number of anchors in $n$ exist is still an intriguing open problem and we hope that our current contributions will inspire future progress.

\clearpage
\printbibliography

\end{document}